\documentclass[12pt]{amsart}
\usepackage{a4wide}
\usepackage{amssymb}
\usepackage{bbm}

\newtheorem{theorem}{Theorem}[section]
\newtheorem{lemma}[theorem]{Lemma}

\theoremstyle{definition}
\newtheorem{definition}[theorem]{Definition}
\newtheorem{remark}[theorem]{Remark}

\numberwithin{equation}{section}

\begin{document}

\title[Portfolio choice under monotone preferences]{Continuous time portfolio choice under monotone preferences with quadratic penalty - stochastic interest rate case}

\author{Jakub Trybu{\l}a}

\address{\noindent Jakub Trybu{\l}a, \newline \indent Institute of Mathematics \newline \indent Faculty of Mathematics and Computer Science \newline \indent  Jagiellonian University in Krakow \newline \indent{\L}ojasiewicza  6 \newline \indent 30-348 Krak{\'o}w, Poland}

\email{jakub.trybula@im.uj.edu.pl}

\author{Dariusz Zawisza}
\address{\noindent Dariusz Zawisza, \newline \indent Institute of Mathematics \newline \indent Faculty of Mathematics and Computer Science \newline \indent  Jagiellonian University in Krakow \newline \indent{\L}ojasiewicza  6 \newline \indent 30-348 Krak{\'o}w, Poland }

\email{dariusz.zawisza@im.uj.edu.pl}



\begin{abstract} This is a follow up of our previous paper - Trybu{\l}a and Zawisza \cite{TryZaw}, where we considered a modification of a monotone mean-variance functional in continuous time in stochastic factor model. In this article we address the problem of optimizing the mentioned functional in a market with a stochastic interest rate. We formulate it as a stochastic differential game problem and use Hamilton-Jacobi-Bellman-Isaacs equations to derive the optimal investment strategy and the value function. 
\end{abstract}

\maketitle

\section{Introduction}\label{sec:0} 
Mean-variance analysis, introduced by Markowitz \cite{Markowitz}, has long been a popular approach to determine the structure and composition of an optimal portfolio. Nevertheless, it is well known that mean-variance functional is not monotone and this is a serious drawback. Namely, Maccheroni et al. \cite{Maccheroni} gave a simple example when an investor with mean-variance preferences may strictly prefer less to more, thus violating one of the most compelling principles of economic rationality. For this reason, they created a new class of monotone preferences that coincide with mean-variance preferences on their domain of monotonicity, but differ where mean-variance preferences fail to be monotone and are therefore not economically meaningful.

A modification of Maccheroni type objective function has been first studied in detail in a dynamic optimization framework  by Trybu{\l}a and Zawisza \cite{TryZaw}. They showed that the solution to the problem of optimizing the mentioned functional in stochastic factor model coincides with the solution to classical Markowitz problem with a suitably chosen risk aversion coefficient. 

In this work we assume that an investor has access to the market, where he can freely invest in a bank account and a risky asset. Moreover, we suppose that the interest rate is given by a solution to a stochastic differential equation driven by one dimensional Brownian motion. The purpose is to describe an optimal financial strategy which an investor can follow in order to maximize his performance criterion which is given by the modification of the monotone mean-variance functional from Trybu{\l}a and Zawisza \cite{TryZaw}. 

The paper is structured in the same manner as mentioned article but the problem here is much harder to solve and results cannot be obtained using direct reasoning from that work. 

For literature review about finite horizon max-min problems we refer to
 Bordigoni et al. \cite{bordigoni}, Hern\'{a}ndez and  Schied \cite{Schied},  Mataramvura and {\O}ksendal \cite{Mataramvura},  {\O}ksendal and Sulem \cite{Oksendal3}, Trybu{\l}a and Zawisza \cite{TryZaw} and Zawisza \cite{Zawisza1}.

For interesting discussion about optimal investment in the presence of stochastic interest rate see for instance Bielecki and  Pliska \cite{Bilecki1} and \cite{Bielecki}, Brennan and  Xia \cite{Brennan}, Korn and Kraft \cite{Korn}, Munk and  Sorensen \cite{Munk} and \cite{Munk2}. It is worth mentioning here the paper of Flor and Larsen \cite{Flor}, where robust utility maximization with homothetic penalty function is considered.

\section{General model description}\label{sec:1} 
 
Let $(\Omega, \mathcal{F}, P)$ be a probability space with a filtration $(\mathcal{F}_{t}, 0 \leq t \leq T)$ possibly enlarged to satisfy usual assumptions and generated by Brownian motion $(W_{t}, 0 \leq t \leq T )$ defined on $(\Omega, \mathcal{F}, P)$. Suppose that an investor has access to the market with a bank account  $(B_{t},s \leq t \leq T)$ and a risky asset $(S_{t},s \leq t \leq T)$. Moreover, we assume that the interest rate is given by a stochastic process $(r_{t},s \leq t \leq T)$ and particulary we will consider the case of the Vasicek model for short rate. Processes mentioned above are solutions of the following system of stochastic differential equations
\begin{equation}\label{model}
\left\{
\begin{aligned}
dB_{t}&=r_{t} B_{t} dt, \\
dS_{t}&=(r_{t} + \lambda(r_{t}) \sigma(r_{t},t)) S_{t}  dt +  \sigma(r_{t},t) S_{t}  dW_{t},   \\
dr_{t}&=\bar{\mu}(r_{t}) dt + \bar{\sigma}(r_{t}) dW_{t},\qquad r_{s}=r>0,
\end{aligned}
\right.
\end{equation}
where the coefficients $\sigma>0$, $\lambda$, $\bar{\mu}$, $\bar{\sigma}$ are continuous functions and they are assumed to satisfy all the required regularity conditions, in order to guarantee that the unique strong solution to (\ref{model}) exists. For notational convenience and in order to ensure that under general conditions there exists a solution to equation (\ref{second_equation}) we avoid putting "t"-dependence in model coefficients. The exception is in the case of function $\sigma$,  since the risky asset $S_{t}$ has a natural interpretation as a price of bond in a short term one factor interest rate models. 

It is still possible to generalize this model by adding additional Brownian motion to the dynamics of $S_{t}$, however we are not able to solve the resulting HJBI equation, which is much extended even in current problem.

The rest of the setting is exactly the same as in Trybu{\l}a and Zawisza \cite{TryZaw}. Namely,
we will consider the class of $P$-equivalent measures
\[
\mathcal{Q}:=\left\{Q\sim P \colon \frac{dQ}{dP} =\mathcal{E} \left( \int \eta_{t}dW_{t}\right)_{T}\; , \quad \eta \in \mathcal{M}\right\},
\]
where $\mathcal{E}(\cdot)_{t}$ denotes the Doleans-Dade exponential and $\mathcal{M}$ is the set of all progressively measurable processes $\eta$ taking values in $\mathbb{R}$, such that
\[
\mathbb{E}\left(\frac{dQ^{\eta}}{dP}\right)^{2} < +\infty \quad \text{ and } \quad \mathbb{E}\left(\frac{dQ^{\eta}}{dP}\right) =1,
\]
where $Q^{\eta}$ denotes the measure determined by $\eta\in \mathcal{M}$. Moreover, let us define additional family of stochastic processes $(Y_{t}^{\eta}, s\leq t\leq T)$ which are given by the stochastic differential equations
\[
dY_{t}^{\eta}=\eta_{t}Y_{t}^{\eta}dW_{t},\qquad Y_{s}^{\eta}=y>0,\qquad \eta\in\mathcal{M}.
\]
Then, notice that
\[
Y_{T}^{\eta}=y\frac{dQ^{\eta}}{dP},\qquad \eta\in\mathcal{M}.
\]

Let $(X^{\pi}_{t}, s \leq t \leq T )$ be the investors wealth process with the following dynamics
\begin{equation}\label{wealth}
 dX_{t}^{\pi}=\left(\pi_{t}\lambda(r_{t})\sigma(r_{t},t)+r_{t}X_{t}^{\pi}\right)dt+\pi_{t}\sigma(r_{t},t)dW_{t},\qquad X_{s}^{\pi}=x>0,
\end{equation}
where $x$ denotes a current wealth of the investor, whereas a control $\pi_{t}$ we can interpret as a part of wealth invested in $S_{t}$. Note that $\pi_{t}$ as well as the portfolio wealth $X_{T}^{\pi}$ are allowed to be negative. 

\begin{definition} \label{defi1} A control (or strategy) $\pi=(\pi_{s}, t \leq s \leq T)$ is admissible on the time interval $[t,T]$, written $\pi \in \mathcal{A}_{x,y,r,t}$, if it satisfies the following assumptions:
\begin{enumerate}
   \item[(i)] $\pi$ is progressively measurable;
   \item[(ii)] unique solution to (\ref{wealth}) exists and
    \[
     \mathbb{E}^{\eta}_{x,y,r,t} \left[\sup_{t \leq s \leq T} |X_{s}^{\pi}|\right] < +\infty\qquad \text{for all}\ \eta \in \mathcal{M},
    \]
		where $\mathbb{E}^{\eta}$ denotes the expectation with respect to measure $Q^{\eta}$.
\end{enumerate}
\end{definition}

\subsection*{Formulation of the problem}\label{sec:2} 

We consider Maccheroni type objective function
\[
J^{\pi,\eta}(x,y,r,t):=
\mathbb{E}^{\eta}_{x,y,r,t}\left[-X_{T}^{\pi}\right]-  y\mathbb{E} \left[\frac{dQ^{\eta}}{dP}\right]^{2}=
\mathbb{E}^{\eta}_{x,y,r,t}\left[-X_{T}^{\pi}-Y_{T}^{\eta}\right].
\]
The investors aim is to
\begin{equation} \label{problem1}
 \text{minimize} \quad \sup_{\eta \in \mathcal{M}}
J^{\pi,\eta}(x,y,r,t)
\end{equation}
over a class of admissible strategies $\mathcal{A}_{x,y,r,t}$.

As usually, the problem (\ref{problem1}) might be considered as a zero-sum stochastic differential game problem.  We are looking for a saddle point  $(\pi^{*},\eta^{*})\in\mathcal{A}_{x,y,r,t}\times\mathcal{M}$ and a value function $V(x,y,r,t)$ such that
\[
 J^{\pi^{*},\eta} (x,y,r,t) \leqslant J^{\pi^{*},\eta^{*}} (x,y,r,t) \leqslant  J^{\pi,\eta^{*}} (x,y,r,t)
\]
and
\[ 
V(x,y,r,t)=J^{\pi^{*},\eta^{*}}(x,y,r,t).
\]

\subsection*{The verification theorem.}\label{sec:3} 

As announced we solve the problem by applying stochastic control theory. Let us remind, that
\begin{equation}\label{model1}
\left\{
\begin{aligned}
dX_{t}^{\pi}&=\left(\pi_{t}\lambda(r_{t})\sigma(r_{t},t)+r_{t}X_{t}^{\pi}\right)dt+\pi_{t}\sigma(r_{t},t)dW_{t},  \\
dY_{t}^{\eta}&=\eta_{t}Y_{t}^{\eta} dW_{t},   \\
dr_{t}&=\bar{\mu}(r_{t}) dt + \bar{\sigma}(r_{t}) dW_{t}.
\end{aligned}
\right.
\end{equation}
It is convenient to consider $Q^{\eta}$-dynamics of system (\ref{model1}). After applaying the Girsanov transformation, we have
\[
\left\{
\begin{aligned}
 dX_{t}^{\pi}&= \left(\pi_{t}\sigma(r_{t},t)(\lambda(r_{t})+\eta_{t})+r_{t}X_{t}^{\pi}\right)dt+\pi_{t}\sigma(r_{t},t)dW_{t}^{\eta}, \\
 dY_{t}^{\eta}&=  \eta_{t}^{2} Y_{t}^{\eta}dt+\eta_{t}Y_{t}^{\eta} dW_{t}^{\eta},\\
 dr_{t}&=\left(\bar{\mu}(r_{t})+\bar{\sigma}(r_{t})\eta_{t}\right) dt + \bar{\sigma}(r_{t}) dW_{t}^{\eta},
\end{aligned}
\right.
\]
where $(W_{t}^{\eta},0\leq t\leq T)$ is $Q^{\eta}$-Brownian motion defined as
\[
dW_{t}^{\eta} = dW_{t}-\eta_{t}dt.
\]
Let $\mathcal{L}^{\pi,\eta}$ be the differential operator given by
\begin{align*}
\mathcal{L}^{\pi,\eta}  V (x,y,r,t) :=&  V_{t}+\left(\pi  \sigma(r,t)\left(\lambda(r)+\eta\right)+rx\right)V_{x}+\eta^{2}yV_{y}+ \left(\bar{\mu}(r)+\bar{\sigma}(r)\eta\right)V_{r} \\
&+\frac{1}{2}\pi^{2}\sigma^{2}(r,t)V_{xx}+\frac{1}{2}\eta^{2}y^{2}V_{yy}+\frac{1}{2}\bar{\sigma}^{2}(r)V_{rr}\\
&+\pi\sigma(r,t)\eta yV_{xy}+\pi\sigma(r,t)\bar{\sigma}(r)V_{xr}+\eta \bar{\sigma}(r)yV_{yr}.
\end{align*}
We can now formulate the Verification Theorem. The proof of this theorem is exactly the same as the proof of analogous
theorem from Trybu{\l}a and Zawisza \cite{TryZaw}, so in this paper we omit it.
\begin{theorem}[Verification Theorem]\label{verification_theorem} Suppose there exists a function 
\[
V \in \mathcal{C}^{2,2,2,1}(\mathbb{R}\times(0,+\infty)\times\mathbb{R} \times [0,T)) \cap \mathcal{C} (\mathbb{R}\times[0,+\infty)\times\mathbb{R} \times [0,T])
\]
and a Markov control 
\[
(\pi^{*},\eta^{*})\in\mathcal{A}_{x,y,r,t}\times\mathcal{M},
\]
such that
\begin{align}
&\mathcal{L}^{\pi^{*}(x,y,r,t),\eta}V(x,y,r,t) \leq 0 \label{first:in1}, \\
&\mathcal{L}^{\pi,\eta^{*}(x,y,r,t)}V(x,y,r,t) \geq 0  \label{second:in1}, \\
&\mathcal{L}^{\pi^{*}(x,y,r,t),\eta^{*}(x,y,r,t)}V(x,y,r,t) = 0\label{third:eq1}, \\
&V(x,y,r,T)=-x-y \label{terminal:cond1}
\end{align}
\flushright for all $\eta \in \mathbb{R}$, $\pi \in \mathbb{R}$,  $(x,y,r,t) \in \mathbb{R}\times(0,+\infty)\times\mathbb{R} \times [0,T) $,
\flushleft
and
\begin{equation} \label{uniform}
 \mathbb{E}_{x,y,r,t}^{\eta} \left[ \sup_{t \leq s \leq T} \left|V(X_{s}^{\pi},Y_{s}^{\eta},r_{s},s)\right|\right] < + \infty 
\end{equation}
\flushright for all  $ (x,y,r,t) \in  \mathbb{R}\times[0,+\infty)\times\mathbb{R} \times [0,T]$,
$\pi \in \mathcal{A}_{x,y,r,t}$, $\eta \in \mathcal{M}$. \flushleft
\medskip

Then
\[J^{\pi^{*},\eta}(x,y,r,t) \leq V(x,y,r,t) \leq J^{\pi,\eta^{*}}(x,y,r,t)\]
\flushright
for all $\pi \in \mathcal{A}_{x,y,r,t}$, $\eta \in \mathcal{M}$,
\flushleft  and
\[ 
V(x,y,r,t)= J^{\pi^{*},\eta^{*}}(x,y,r,t). 
\]
\end{theorem}

\subsection*{Solution to the minimax problem}\label{sec:4} 

To find the saddle point we start with analizing a Hamilton-Jacobi-Bellman-Isaacs equation
\begin{equation} \label{upper_isaacs}
\min_{\pi \in \mathbb{R}} \max_{\eta \in \mathbb{R}} \mathcal{L}^{\pi,\eta} V(x,y,r,t)=0,
\end{equation}
i.e.
\begin{align*}
V_{t}+rxV_{x}+&\bar{\mu}(r)V_{r}+\frac{1}{2}\bar{\sigma}^{2}(r)V_{rr}\\
+\min_{\pi\in\mathbb{R}}\max_{\eta\in\mathbb{R}}&\biggl\lbrace  \pi  \sigma(r,t)\left(\lambda(r)+\eta\right)V_{x}+\eta^{2}yV_{y}+ \bar{\sigma}(r)\eta V_{r}+\frac{1}{2}\pi^{2}\sigma^{2}(r)V_{xx}\\
&+\frac{1}{2}\eta^{2}y^{2}V_{yy}+\pi\sigma(r,t)\eta y V_{xy}+\pi\sigma(r,t)\bar{\sigma}(r)V_{xr}+\eta \bar{\sigma}(r)yV_{yr}\biggl\rbrace=0.
\end{align*}
We expect $V(x,y,r,t)$ to be of the form
\begin{equation} \label{given}
V(x,y,r,t)=H(r,t)x+G(r,t)y,
\end{equation}
where
\[
H(r,T)=-1\quad\text{and}\quad G(r,T)=-1.
\]
Then we have
\begin{align*}
xH_{t}+yG_{t}&+rxH+\bar{\mu}(r)\left(xH_{r}+yG_{r}\right)+\frac{1}{2}\bar{\sigma}^{2}(r)\left(xH_{rr}+yG_{rr}\right)\\
+\min_{\pi\in\mathbb{R}}\max_{\eta\in\mathbb{R}}&\biggl\lbrace \pi\sigma(r,t)(\lambda(r)+\eta)H+\eta^{2}yG\\
&+\bar{\sigma}(r)\left(\eta x +\pi\sigma(r,t)\right)H_{r}+2\eta \bar{\sigma}(r)yG_{r}\biggl\rbrace=0.
\end{align*}
The maximum over $\eta$ is attained at $\eta^{*}(\pi)$, where
\[
\eta^{*}(\pi)=-\frac{\sigma(r,t)H}{2yG}\pi-\frac{\bar{\sigma}(r)(xH_{r}+2yG_{r})}{2yG}.
\]
For $\eta^{*}(\pi)$ our equation is of the form
\begin{align}
xH_{t}+&yG_{t}+rxH+\bar{\mu}(r)\left(xH_{r}+yG_{r}\right)+\frac{1}{2}\bar{\sigma}^{2}(r)\left(xH_{rr}+yG_{rr}\right)\notag \\
+\min_{\pi\in\mathbb{R}}&\biggl\lbrace \pi\sigma(r,t)(\lambda(r)+\eta^{*}(\pi))H+\left(\eta^{*}(\pi)\right)^{2}yG\label{eq2}\\
&+\bar{\sigma}(r)\left(\eta^{*}(\pi) x+\pi\sigma(r,t)\right)H_{r}+2\eta^{*}(\pi) \bar{\sigma}(r)yG_{r}\biggl\rbrace=0.\notag 
\end{align}
The minimum over $\pi$ is attained at
\begin{equation}\label{pistar}
\pi^{*}=2yG\left[\frac{\lambda(r)}{\sigma(r,t)}\frac{1}{H}+\frac{\bar{\sigma}(r)}{\sigma(r,t)}\left(\frac{H_{r}}{H^{2}}-\frac{G_{r}}{GH}\right)\right]-x\frac{\bar{\sigma}(r)}{\sigma(r,t)}\frac{H_{r}}{H}.
\end{equation}
It is worth to notice here that 
\begin{equation}\label{eta_pi}
\eta^{*}(\pi^{*})=-\lambda(r)-\bar{\sigma}(r)\frac{H_{r}}{H},
\end{equation}
so the saddle point candidate 
\begin{equation}\label{saddle_point_candidate}
\left(\pi^{*},\eta^{*}(\pi^{*})\right)
\end{equation}
looks as follows
\begin{align*}
\pi^{*}&=2yG\left[\frac{\lambda(r)}{\sigma(r,t)}\frac{1}{H}+\frac{\bar{\sigma}(r)}{\sigma(r,t)}\left(\frac{H_{r}}{H^{2}}-\frac{G_{r}}{GH}\right)\right]-x\frac{\bar{\sigma}(r)}{\sigma(r,t)}\frac{H_{r}}{H}  ,\\
\eta^{*}(\pi^{*})& =-\lambda(r)-\bar{\sigma}(r)\frac{H_{r}}{H}.
\end{align*}
Now we substitute (\ref{pistar}) into (\ref{eq2}) and get the final equation of the form
\begin{align*}
x\biggl[H_{t}+rH&+\left(\bar{\mu}(r)-\bar{\sigma}(r)\lambda(r)\right)H_{r}+\frac{1}{2}\bar{\sigma}^{2}(r)H_{rr}-\bar{\sigma}^{2}(r)\frac{H^{2}_{r}}{H}\biggr] \\
+y\biggl[&G_{t}+\frac{1}{2}\bar{\sigma}^{2}(r)G_{rr}+\left(\lambda(r)+\bar{\sigma}(r)\frac{H_{r}}{H}\right)^{2}G  \\
&+\left(\bar{\mu}(r)-2\bar{\sigma}(r)\lambda(r)-2\bar{\sigma}^{2}(r)\frac{H_{r}}{H}\right)G_{r}\biggr]=0.
\end{align*}
Therefore, instead of solving completely nonlinear equations it is sufficient to find a classical (class $\mathcal{C}^{2,1}$) unique solutions for two semilinear equations:
\begin{equation}\label{first_equation}
H_{t}+rH+\left(\bar{\mu}(r)-\bar{\sigma}(r)\lambda(r)\right)H_{r}+\frac{1}{2}\bar{\sigma}^{2}(r)H_{rr}-\bar{\sigma}^{2}(r)\frac{H^{2}_{r}}{H}=0,
\end{equation}
with terminal condition $H(r,T)=-1$ and 
\begin{align}
\begin{split}\label{second_equation}
    &G_{t}+\frac{1}{2}\bar{\sigma}^{2}(r)G_{rr}+\left(\lambda(r)+\bar{\sigma}(r)\frac{H_{r}}{H}\right)^{2}G\\
         &+\left(\bar{\mu}(r)-2\bar{\sigma}(r)\lambda(r)-2\bar{\sigma}^{2}(r)\frac{H_{r}}{H}\right)G_{r}=0,
\end{split}
\end{align}
with terminal condition $G(r,T)=-1$. We will back to these two equations in Section \ref{sec:5}.

Now we need to prove two usefull lemmas. 

\begin{lemma}\label{minimax_equalities_lemma}
Suppose that function $V  \in \mathcal{C}^{2,2,2,1}(\mathbb{R}\times(0,+\infty)\times\mathbb{R} \times [0,T))$ given by (\ref{given}) is a classical unique solution to (\ref{upper_isaacs}). Moreover, let $\left(\pi^{*},\eta^{*}(\pi^{*})\right)\in\mathcal{A}_{x,y,r,t}\times\mathcal{M}$ be determined using (\ref{saddle_point_candidate}). Then conditions (\ref{first:in1}) - (\ref{terminal:cond1}) of Theorem \ref{verification_theorem} are satisfied.
\end{lemma}
\begin{proof}
We already know that 
\[
\max_{\eta \in \mathbb{R}}  \mathcal{L}^{\pi^{*},\eta} V(x,y,r,t) = 0,\qquad \mathcal{L}^{\pi^{*},\eta^{*}} V(x,y,r,t) = 0
\]
and 
\[
  V(x,y,r,T)=-x-y,
\]
which confirms (\ref{first:in1}), (\ref{third:eq1}) and (\ref{terminal:cond1}).

To prove (\ref{second:in1}) it is sufficient to use (\ref{eq2}) and (\ref{eta_pi}) and simply verify that 
\[
 \min_{\pi \in \mathbb{R}} \mathcal{L}^{\pi,\eta^{*}(\pi^{*})} V(x,y,r,t)= 0.
\]
\end{proof}

The second lemma will be helpful in Section~\ref{sec:6} to prove the main result.

\begin{lemma} \label{lem_reduction}
Suppose that initial conditions $(x_{0},y_{0},r_{0},t_{0})$ are fixed, the saddle point
\[
\left(\pi^{*},\eta^{*}(\pi^{*})\right)\in \mathcal{A}_{x_{0},y_{0},r_{0},t_{0}}\times\mathcal{M}
\]
is given by (\ref{saddle_point_candidate}), $H(r,t)$ and $G(r,t)$ are a classical unique solutions to equations (\ref{first_equation}) and (\ref{second_equation}) respectively. Then
\[
2G(r_{t},t)Y^{\eta^{*}}_{t}=2G(r_{0},t_{0})y_{0}+H(r_{0},t_{0})x_{0}-H(r_{t},t)X^{\pi^{*}}_{t},\qquad \forall t\in [t_{0},T].
\]
\end{lemma}
\begin{proof}
It is sufficient to prove only that 
\[
d\left(H(r_{t},t)X^{\pi^{*}}_{t}\right)=d\left(-2G(r_{t},t)Y^{\eta^{*}}_{t}\right).
\]
First of all, note that for saddle point given by (\ref{saddle_point_candidate}) system of equations (\ref{model1}) is of the form
\begin{align*}
dX_{t}^{\pi^{*}}=&\left\{2Y_{t}^{\eta^{*}}G\left[\frac{\lambda^{2}(r_{t})}{H}+\bar{\sigma}(r_{t})\lambda(r_{t})\left(\frac{H_{r}}{H^{2}}-\frac{G_{r}}{GH}\right)\right]-X_{t}^{\pi^{*}}\lambda(r_{t})\bar{\sigma}(r_{t})\frac{H_{r}}{H}+r_{t}X_{t}^{\pi^{*}}\right\}dt\notag \\
&+\left\{2Y_{t}^{\eta^{*}}G\left[\frac{\lambda(r_{t})}{H}+\bar{\sigma}(r_{t})\left(\frac{H_{r}}{H^{2}}-\frac{G_{r}}{GH}\right)\right]-X_{t}^{\pi^{*}}\bar{\sigma}(r_{t})\frac{H_{r}}{H}\right\}dW_{t}
\end{align*}
and
\[
dY^{\eta^{*}}_{t} =-\left\{\lambda(r_{t})+\bar{\sigma}(r_{t})\frac{H_{r}}{H}\right\}Y^{\eta^{*}}_{t} dW_{t}.
\]
Using (\ref{first_equation}) we can verify that
\[
dH(r_{t},t)=\left[-r_{t}H(r_{t},t)+\bar{\sigma}(r_{t})\lambda(r_{t})H_{r}(r_{t},t)+\bar{\sigma}^{2}(r_{t})\frac{H^{2}_{r}(r_{t},t)}{H(r_{t},t)}\right]dt+\bar{\sigma}(r_{t})H_{r}(r_{t},t)dW_{t}. 
\]
Moreover, we have 
\[
d\left(H(r_{t},t)X_{t}^{\pi^{*}}\right)=H(r_{t},t)dX_{t}^{\pi^{*}}+X_{t}^{\pi^{*}}dH(r_{t},t)+dH(r_{t},t)dX_{t}^{\pi^{*}},
\]
so substituting the appropriate dynamics to above equation we get
\begin{align}\label{aa:1}
d(H(r_{t},t)X_{t}^{\pi^{*}})=2Y_{t}^{\eta^{*}}&\left\{\left[\left(\lambda(r_{t})+\bar{\sigma}(r_{t})\frac{H_{r}(r_{t},t)}{H(r_{t},t)}\right)^{2}G(r_{t},t)\right.\right.\notag \\
&\left.\left.-\bar{\sigma}(r_{t})\left(\lambda(r_{t})+\bar{\sigma}(r_{t})\frac{H_{r}(r_{t},t)}{H(r_{t},t)}\right)G_{r}(r_{t},t)\right]\right. dt\\
&+\left.\left[\left(\lambda(r_{t})+\bar{\sigma}(r_{t})\frac{H_{r}(r_{t},t)}{H(r_{t},t)}\right)G(r_{t},t)-\bar{\sigma}(r_{t})G_{r}(r_{t},t)\right]dW_{t}\right\}\notag .
\end{align}
Now using (\ref{second_equation}) we can verify that
\begin{align*}
dG(r_{t},t)=&\left[2\bar{\sigma}(r_{t})\left(\lambda(r_{t})+\bar{\sigma}(r_{t})\frac{H_{r}(r_{t},t)}{H(r_{t},t)}\right)G_{r}(r_{t},t)\right. \\
&\left.-\left(\lambda(r_{t})+\bar{\sigma}(r_{t})\frac{H_{r}(r_{t},t)}{H(r_{t},t)}\right)^{2}G(r_{t},t)\right]dt+\bar{\sigma}(r_{t})G_{r}(r_{t},t)dW_{t}.
\end{align*}
Moreover, we have 
\[
d\left(-2G(r_{t},t)Y^{\eta^{*}}_{t}\right)=-2G(r_{t},t)dY^{\eta^{*}}_{t}-2Y^{\eta^{*}}_{t}dG(r_{t},t)-2dG(r_{t},t)dY^{\eta^{*}}_{t},
\]
so substituting the appropriate dynamics to above equation we get the right hand side of (\ref{aa:1}).
\end{proof}

\begin{remark}\label{reduction2}
Note that process $(Y^{\eta^{*}}_{t},t_{0}\leq t\leq T)$ is not directly observable but fortunately the above lemma ensures that for fixed initial conditions $(x_{0},y_{0},r_{0},t_{0})$ instead of Markov strategy 
\[
\pi^{*}=2Y_{t}^{\eta^{*}}G\left[\frac{\lambda(r_{t})}{\sigma(r_{t},t)}\frac{1}{H}+\frac{\bar{\sigma}(r_{t})}{\sigma(r_{t},t)}\left(\frac{H_{r}}{H^{2}}-\frac{G_{r}}{GH}\right)\right]-X_{t}^{\pi^{*}}\frac{\bar{\sigma}(r_{t})}{\sigma(r_{t},t)}\frac{H_{r}}{H},
\]
we can use
\begin{align*}
\hat{\pi}^{*}=&\bigl[2G(r_{0},t_{0})y_{0}+H(r_{0},t_{0})x_{0}-H(r_{t},t) X_{t}^{\pi^{*}}\bigr] \\
&*\left[\frac{\lambda(r_{t})}{\sigma(r_{t},t)}\frac{1}{H}+\frac{\bar{\sigma}(r_{t})}{\sigma(r_{t},t)}\left(\frac{H_{r}}{H^{2}}-\frac{G_{r}}{GH}\right)\right]-X_{t}^{\pi^{*}}\frac{\bar{\sigma}(r_{t})}{\sigma(r_{t},t)}\frac{H_{r}}{H}.
\end{align*}
\end{remark}

\begin{section}{Classical smooth solutions to resulting equations}\label{sec:5} 

In this section we give a set of assumptions to ensure existence of a classical (class $\mathcal{C}^{2,1}$) unique solutions to equations $(\ref{first_equation})$ and $(\ref{second_equation})$ with appropriate terminal conditions.

To solve equation (\ref{first_equation}) with boundary condition $H(r,T)=-1$ the following ansatz is made 
\[
H(r,t)=-\Gamma^{a}(r,t),\qquad\text{where}\qquad \Gamma(r,T)=1\qquad\text{and}\qquad \Gamma(r,t)>0,
\]
to obtain
\[
\Gamma_{t}+\frac{r}{a}\Gamma+\left(\bar{\mu}(r)-\bar{\sigma}(r)\lambda(r)\right)\Gamma_{r}+\frac{1}{2}\bar{\sigma}^{2}(r)\Gamma_{rr}+\left[\frac{1}{2}(a-1)-a\right]\bar{\sigma}^{2}(r)\frac{\Gamma^{2}_{r}}{\Gamma}=0.
\]
Note that for $a=-1$ we have
\begin{equation} \label{eqgamma}
\Gamma_{t}-r\Gamma+\left(\bar{\mu}(r)-\bar{\sigma}(r)\lambda(r)\right)\Gamma_{r}+\frac{1}{2}\bar{\sigma}^{2}(r)\Gamma_{rr}=0.
\end{equation}
Nevertheless, it is convenient to make another substitution. Namely, if we substitute
\[
\Gamma(r,t)=e^{-r(T-t)}F(r,t),\qquad\text{where}\qquad F(r,T)=1\qquad\text{and}\qquad F(r,t)>0,
\]
we get the following equation
\begin{equation}\label{final_final}
F_{t}+\left[\frac{1}{2}\bar{\sigma}^{2}(r)\left(T-t\right)^{2}-\left(T-t\right)\left(\bar{\mu}(r)-\bar{\sigma}(r)\lambda(r)\right)\right]F
\end{equation}
\[
+\left[\bar{\mu}(r)-\bar{\sigma}(r)\lambda(r)-\bar{\sigma}^{2}(r)(T-t)\right]F_{r}+\frac{1}{2}\bar{\sigma}^{2}(r)F_{rr}=0.
\]
Now we are going to prove existence of a classical (class $\mathcal{C}^{2,1}$) unique solution to equation (\ref{final_final}) with boundary condition $F(r,T)=1$.

\begin{remark}\label{important_remark} 
From Theorem 4.6, Chapter 6 of Friedman \cite{Friedman}, it follows that if $\bar{\mu}$, $\bar{\sigma}$ and $\bar{\sigma}\cdot\lambda$ are Lipschitz continuous and bounded, $\bar{\sigma}^{2}>\varepsilon>0$, then there exists a classical (class $\mathcal{C}^{2,1}(\mathbb{R} \times [0,T)) \cap\mathcal{C}(\mathbb{R} \times [0,T]) $) unique solution $F$ to equation (\ref{final_final}) which is bounded (see in addition estimate 4.12, Chapter 6 of Friedman \cite{Friedman}).  It is well known that such solution  satisfies the Feynman-Kac representation:
\[
F(r,t)=\mathbb{E}_{r,t}^{\tilde{P}}\left[ \exp\left\{\int_{t}^{T} \phi(\tilde{r}_{s},s) ds\right\}\right], 
\]
where
\[
 \phi(r,t):=\frac{1}{2}\bar{\sigma}^{2}(r)\left(T-t\right)^{2}-\left(T-t\right)\left(\bar{\mu}(r)-\bar{\sigma}(r)\lambda(r)\right),
\]
$(\tilde{W}_{s}, t \leq s \leq T)$ is a Brownian motion with respect to $\tilde{P}$ and
\[
d\tilde{r}_{s}=\left[\bar{\mu}(\tilde{r}_{s})-\bar{\sigma}(\tilde{r}_{s})\lambda(\tilde{r}_{s})-\bar{\sigma}^{2}(\tilde{r}_{s})(T-t)\right]ds +\bar{\sigma}(\tilde{r}_{s}) d\tilde{W}_{s},\qquad \tilde{r}_{t}=r.
\]
Moreover, note that $\phi$ is bounded and Lipschitz continuous function. It means that $F$ is bounded and bounded away from zero.
\end{remark}

To prove other properties of $F$ we need two more lemmas. 

\begin{lemma} \label{uniform-r}
 Suppose that process $(h_{t}, 0 \leq t \leq T)$ with deterministic starting point is given by
\[
 dh(t)=\zeta_{1}(t)dt + \zeta_{2}(t) d\tilde{W}_{t},  
\]
where $\zeta_{1}$ and $\zeta_{2}$ are  bounded stochastic processes and $(\tilde{W}_{t}, 0 \leq t \leq T)$ is a Brownian motion with respect to $\tilde{P}$. Then  
\[
\mathbb{E}^{\tilde{P}} \left[\sup_{0 \leq t \leq T} \exp \left\{ \int_{0}^{t}h(s) ds \right\}\right] < +\infty.
\]
\end{lemma}

\begin{proof}
Note that 
\[
d(t h(t))= (h(t)+t\zeta_{1}(t))dt + t\zeta_{2}(t) d\tilde{W}_{t},
\]
which can be rewritten into 
\[
th(t)=\int_{0}^{t}h(s)ds+\int_{0}^{t} s\zeta_{1}(s) ds +\int_{0}^{t} s\zeta_{2}(s) d\tilde{W}_{s}. 
\]
It means that
\[
\mathbb{E}^{\tilde{P}} \left[\sup_{0 \leq t \leq T} \exp \left\{ \int_{0}^{t}h(s) ds \right\}\right] = \mathbb{E}^{\tilde{P}} \left[\sup_{0 \leq t \leq T}  \exp \left\{ th(t)-\int_{0}^{t} s\zeta_{1}(s) ds -\int_{0}^{t} s\zeta_{2}(s) d\tilde{W}_{s}  \right\}\right].
\]
Moreover, we have
\[
\sup_{0 \leq t \leq T}  \exp \left\{ th(t)-\int_{0}^{t} s\zeta_{1}(s) ds-\int_{0}^{t} s\zeta_{2}(s) d\tilde{W}_{s} \right\}
\]
\begin{align*}
\leq&\sup_{0 \leq t \leq T} \mathbbm{1}_{\{h(t)<0\}} \exp \left\{ th(t)-\int_{0}^{t} s\zeta_{1}(s) ds -\int_{0}^{t} s\zeta_{2}(s) d\tilde{W}_{s}  \right\}\\
&+ \sup_{0 \leq t \leq T} \mathbbm{1}_{\{h(t) \geq 0\}} \exp \left\{ th(t)-\int_{0}^{t} s\zeta_{1}(s) ds -\int_{0}^{t} s\zeta_{2}(s) d\tilde{W}_{s}\right\}\\
\leq &\sup_{0 \leq t \leq T}  \exp \left\{-\int_{0}^{t} s\zeta_{1}(s) ds -\int_{0}^{t} s\zeta_{2}(s) d\tilde{W}_{s}  \right\} \\
&+ \sup_{0 \leq t \leq T}  \exp \left\{ Th(t)-\int_{0}^{t} s\zeta_{1}(s) ds -\int_{0}^{t} s\zeta_{2}(s) d\tilde{W}_{s}\right\}.
\end{align*}
This concludes the proof since both processes under supremum are solutions to linear equations with bounded coefficients.
\end{proof}

\begin{lemma} \label{derivative}
Suppose $\bar{\mu}$, $\bar{\sigma}$ and $\bar{\sigma}\cdot\lambda$ are Lipschitz continuous and bounded, $\bar{\sigma}^{2}> \varepsilon > 0$ and $F$ is a bounded solution to equation  (\ref{final_final}). Then the first and the second $r$-derivative of $F$ are bounded.
\end{lemma}

\begin{proof} 
To get a bound for $F_{r}$ it is sufficient to estimate the Lipschitz constant. First of all, note that for $r_{1},r_{2} \in (- \infty, a]$ there exists $L_{a}>0$ such that 
\begin{equation}\label{property}
|e^{r_{1}}-e^{r_{2}}| \leq L_{a}|r_{1}-r_{2}|.
\end{equation}
Secondly using (\ref{property}) and the fact that from Remark \ref{important_remark} function $\phi(r,t)$ is bounded and Lipschitz continuous we obtain existence of $L>0$ that
\begin{align*}
|F(r,t)-F(\bar{r},t)| \leq&  L\mathbb{E}^{\tilde{P}}\left[ \int_{t}^{T}\left|\tilde{r}_{s}(r,t)-\tilde{r}_{s}(\bar{r},t)\right| ds\right]\\ 
\leq& L T \mathbb{E}^{\tilde{P}} \left[\sup_{t \leq s \leq T} \left|\tilde{r}_{s}(r,t)-\tilde{r}_{s}(\bar{r},t)\right|\right],
\end{align*}
where from notational covenience we wrote $\mathbb{E}^{\tilde{P}}f(\tilde{r}_{s}(r,t))$ instead of $\mathbb{E}^{\tilde{P}}_{r,t}f(\tilde{r}_{s})$. Now, it is well known (Theorem 1.3.16 from Pham \cite{Pham2}) there exists $C_{T}>0$ such that 
\[
\mathbb{E}^{\tilde{P}} \left[\sup_{t \leq s \leq T} \left|\tilde{r}_{s}(r,t)-\tilde{r}_{s}(\bar{r},t)\right|\right] \leq C_{T}|r-\bar{r}|.
\]

To prove boundness of $F_{rr}$ we first estimate $F_{t}$ and use the fact that $F$ is a solution to equation (\ref{final_final}). Let us recall that
\[
\Gamma(r,t)=e^{-r(T-t)}F(r,t)
\]
is a solution to (\ref{eqgamma}). Suppose that $t \leq T$ is fixed and define function
\[
 v(r,k)=\Gamma(r,k+T-t),\qquad k\in[0,t].
\]
It is straightforward that $v$ is a bounded solution to (\ref{eqgamma}) but with terminal condition $v(r,t)=1$. Lemma \ref{uniform-r} ensures that Feynman-Kac representation is possible i.e.
\[
 v(r,0)=\mathbb{E}_{r,0}^{\tilde{P}}\left[ \exp\left\{- \int_{0}^{t}\tilde{r}_{s} ds \right\}\right]=
e^{-rt} \mathbb{E}_{r,0}^{\tilde{P}}\left[ \exp\left\{- \int_{0}^{t}g(s) ds \right\}\right], 
\]
where $(\tilde{W}_{s}, 0 \leq s \leq t)$ is a Brownian motion with respect to $\tilde{P}$,
\[
d\tilde{r}_{s}=\left(\bar{\mu}(\tilde{r}_{s})-\bar{\sigma}(\tilde{r}_{s})\lambda(\tilde{r}_{s})\right)dt +\bar{\sigma}(\tilde{r}_{s}) d\tilde{W}_{s},\qquad \tilde{r}_{0}=r
\]
 and 
 \[g(t)= \int_{0}^{t} \left(\bar{\mu}(\tilde{r}_{s})-\bar{\sigma}(\tilde{r}_{s})\lambda(\tilde{r}_{s})\right)ds + \int_{0}^{t}\bar{\sigma}(\tilde{r}_{s}) d\tilde{W}_{s}.
 \] 
Now, note that 
\[
F(r,T-t)=e^{rt}v(r,0),
\]
so after differentiating $F$ with respect to $t$ we get
\[
|F_{t}(r,T-t)| \leq  \mathbb{E}_{r,0}^{\tilde{P}} \left[ |g(t)| \exp\left\{-\int_{0}^{t} g(s) ds\right\}\right].
\]
Differentiation is possible since 
\[
\mathbb{E}_{r,0}^{\tilde{P}} \left[\sup_{0 \leq t \leq T} \left( |g(t)| \exp\left\{-\int_{0}^{t} g(s) ds\right\}\right)\right]
\]
\[
\leq \frac{1}{2} \mathbb{E}_{r,0}^{\tilde{P}} \left[\sup_{0 \leq t \leq T}  |g(t)|^{2}\right] + \frac{1}{2}\mathbb{E}_{r,0}^{\tilde{P}} \left[\sup_{0 \leq t \leq T}  \exp\left\{-2 \int_{0}^{t} g(s) ds\right\}\right]
\]
and from Lemma \ref{uniform-r} we know that 
\[
\mathbb{E}_{r,0}^{\tilde{P}} \left[\sup_{0 \leq t \leq T}\exp\left\{-2 \int_{0}^{t} g(s) ds\right\}\right] < +\infty.
\]
\end{proof}

Now we are ready to consider the second equation (\ref{second_equation}). Note that  
\begin{equation}\label{fraction}
\frac{H_{r}}{H}=-\frac{\Gamma_{r}}{\Gamma}=(T-t)-\frac{F_{r}}{F}
\end{equation}
and equation (\ref{second_equation}) has the following form 
\begin{equation}\label{second_equation_2}
G_{t}+\frac{1}{2}\bar{\sigma}^{2}(r)G_{rr}+\left(\lambda(r)+\bar{\sigma}(r)\left[(T-t) - \frac{F_{r}}{F}\right]\right)^{2}G
\end{equation}
\[
+\left(\bar{\mu}(r)-2\bar{\sigma}(r)\lambda(r)-2\bar{\sigma}^{2}(r)\left[(T-t) - \frac{F_{r}}{F}\right]\right)G_{r}=0.
\]

\begin{remark}\label{im_re} 
Under conditions of Lemma \ref{derivative} we get boundness and Lipschitz continuity of (\ref{fraction}). This ensures that  there exists  a classical (class $\mathcal{C}^{2,1}(\mathbb{R} \times [0,T)) \cap\mathcal{C}(\mathbb{R} \times [0,T]) $) unique solution to \eqref{second_equation_2}, which satisfy Feynman-Kac representation:
\[
G(r,t)=\mathbb{E}_{r,t}^{\hat{P}}\left[ \exp\left\{\int_{t}^{T} \psi(\hat{r}_{s},s) ds\right\}\right], 
\]
where 
\[
 \psi(r,t)=\left(\lambda(r)+\bar{\sigma}(r)\left[(T-t) - \frac{F_{r}}{F}\right]\right)^{2},
\]
$(\hat{W}_{s}, t \leq s \leq T)$ is a Brownian motion with respect to $\hat{P}$ and 
\[
d\hat{r}_{s}=\left(\bar{\mu}(\hat{r}_{s})-2\bar{\sigma}(\hat{r}_{s})\lambda(\hat{r}_{s})-2\bar{\sigma}^{2}(\hat{r}_{s})\left[(T-s) - \frac{F_{r}}{F}\right]\right)ds +\bar{\sigma}(\hat{r}_{s}) d\hat{W}_{s},\qquad \hat{r}_{t}=r.
\]
At the end it is worth noticing that since  $\psi$ is Lipschitz continuous in $r$ (uniformly wrt. $t$) and bounded, then function $G$ is bounded, bounded away from zero and first $r$-derivative of $G$ is bounded (see proof of Lemma \ref{derivative}).  
\end{remark} 
\end{section}

\section{Final solution}\label{sec:6}

\begin{theorem} \label{final}
Suppose $\bar{\sigma}^{2}>\varepsilon>0$ and $\bar{\mu}$, $\bar{\sigma}$ and $\bar{\sigma}\cdot\lambda$ are Lipschitz continuous and bounded. Then for each initial conditions $(x_{0},y_{0},r_{0},t_{0})$ there exists a Markov saddle point
\begin{equation}\label{mmmarkov}
(\pi^{*},\eta^{*})\in \mathcal{A}_{x_{0},y_{0},r_{0},t_{0}} \times\mathcal{M}
\end{equation}
for problem (\ref{problem1}) such that 
\begin{align*}
\pi^{*}&=2Y_{t}^{\eta^{*}}G\left[\frac{\lambda(r_{t})}{\sigma(r_{t},t)}\frac{1}{H}+\frac{\bar{\sigma}(r_{t})}{\sigma(r_{t},t)}\left(\frac{H_{r}}{H^{2}}-\frac{G_{r}}{GH}\right)\right]-X_{t}^{\pi^{*}}\frac{\bar{\sigma}(r_{t})}{\sigma(r_{t},t)}\frac{H_{r}}{H},\\
\eta^{*}&=-\lambda(r_{t})-\bar{\sigma}(r_{t})\frac{H_{r}}{H},
\end{align*}
where $G$ and $H$ are a classical unique solutions to (\ref{first_equation}) and (\ref{second_equation}) respectively.
\end{theorem}

\begin{proof}
It follows from Remark \ref{important_remark}, Lemma \ref{derivative} and Remark \ref{im_re} that there exist a classical unique solutions to (\ref{first_equation}) and (\ref{second_equation}), which are bounded, bounded away from zero and have first $r$-derivatives bounded. 

If we set
\[
 V(x,y,r,t):=H(r,t)x+G(r,t)y,
\]
then it is sufficient to check whether function $V$ and Markov saddle point (\ref{mmmarkov}) satisfy all conditions of the Verification Theorem. Due to calculations (\ref{upper_isaacs}) - (\ref{second_equation}) and Lemma \ref{minimax_equalities_lemma}, conditions (\ref{first:in1}) - (\ref{terminal:cond1}) are fulfilled. Now, we only have to prove that $(\pi^{*},\eta^{*})$ belongs to the set $\mathcal{A}_{x_{0},y_{0},r_{0},t_{0}} \times\mathcal{M}$ and condition (\ref{uniform}) holds, so we have to show that for any $\eta \in \mathcal{M}$
 \begin{equation} \label{1}
 \mathbb{E}_{x_{0},y_{0},r_{0},t_{0}}^{\eta} \left[ \sup_{t_{0} \leq s \leq T} \left|V(X_{s}^{\pi^{*}},Y_{s}^{\eta^{*}},r_{s},s)\right|\right] < + \infty
\end{equation}
and
\begin{equation} \label{2}
 \mathbb{E}^{\eta}_{x_{0},y_{0},r_{0},t_{0}} \left[\sup_{t_{0} \leq s \leq T} |X_{s}^{\pi^{*}}|\right] < +\infty.
\end{equation}
First of all, let us remind that $Y^{\eta^{*}}_{t}$ is a solution to equation
\begin{equation}\label{aaabbbccc}
dY^{\eta^{*}}_{t} =-\left\{\lambda(r_{t})+\bar{\sigma}(r_{t})\frac{H_{r}}{H}\right\}Y^{\eta^{*}}_{t} dW_{t}.
\end{equation}
Since $G$ is bounded and (\ref{aaabbbccc}) is a linear stochastic differential equation with bounded coefficients, we have
\begin{equation}\label{4}
 \mathbb{E}_{x_{0},y_{0},r_{0},t_{0}}^{\eta} \left[ \sup_{t_{0} \leq s \leq T} \left|G(r_{s},s)Y_{s}^{\eta^{*}}\right|\right]
\end{equation}
\[
 \leq \sqrt{\mathbb{E}^{P}\left(\frac{dQ^{\eta}}{dP}\right)^{2}}\sqrt{\mathbb{E}_{x_{0},y_{0},r_{0},t_{0}}^{P}\sup_{t_{0} \leq s \leq T} \left|G(r_{s},s)Y_{s}^{\eta^{*}}\right|^2} < + \infty\qquad  \forall\eta \in \mathcal{M}.
\]

To prove the same with $H(r_{t},t)X_{t}^{\pi^{*}}$ we use Lemma \ref{lem_reduction}. Namely, let us remind that for fixed initial conditions $(x_{0},y_{0},r_{0},t_{0})$ we have
\begin{equation}\label{3}
H(r_{t},t)X^{\pi^{*}}_{t}=2G(r_{0},t_{0})y_{0}+H(r_{0},t_{0})x_{0}-2G(r_{t},t)Y^{\eta^{*}}_{t}.
\end{equation}
Since $G$ and $Y^{\eta^{*}}_{t}$ satisfy (\ref{4}), then 
\[
 \mathbb{E}_{x_{0},y_{0},r_{0},t_{0}}^{\eta} \left[ \sup_{t_{0} \leq s \leq T} \left|H(r_{s},s)X_{s}^{\eta^{*}}\right|\right]< + \infty\qquad \forall\eta \in \mathcal{M}
\]
and condition (\ref{1}) holds.

To prove (\ref{2}) first we show that for any deterministic continuous function $w(t)$ and for any $\eta \in  \mathcal{M}$ we have
\[
 \mathbb{E}_{r_{0},t_{0}}^{\eta}\left[\sup_{t_{0} \leq t \leq T} e^{w(t)r_{t}}\right]  < + \infty,
\]
where $(r_{t},t_{0}\leq t\leq T)$ is a stochastic process given by
\[
dr_{t}=\bar{\mu}(r_{t}) dt + \bar{\sigma}(r_{t}) dW_{t}.
\]
It is easy to see, that
\[
r_{t}=r_{0}+\int_{t_{0}}^{t}\bar{\mu}(r_{s}) ds+\int_{t_{0}}^{t}\bar{\sigma}(r_{s}) dW_{s}
\]
and consequently, since $e^{r_{t}}$ is a solution to linear stochastic equation with bounded coefficients, we have
\begin{equation}\label{one_1}
\mathbb{E}_{r_{0},t_{0}}^{P} \left[\sup_{t_{0} \leq t \leq T} e^{r_{t}}\right]<+\infty.
\end{equation}
Now, observe that
\begin{align*}
 \sup_{t_{0} \leq t \leq T} e^{w(t)r_{t}}  &\leq \sup_{t_{0} \leq t \leq T} \mathbbm{1}_{\{ r_{t}<0\}}  e^{w(t)r_{t}} + \sup_{t_{0} \leq t \leq T} \mathbbm{1}_{\{ r_{t} \geq 0\}}  e^{w(t)r_{t}}  \\ &\leq \sup_{t_{0} \leq t \leq T} \mathbbm{1}_{\{ r_{t}<0\}}  e^{\underline{w}r_{t}} + \sup_{t_{0} \leq t \leq T} \mathbbm{1}_{\{ r_{t} \geq 0\}}  e^{\overline{w}r_{t}} \\
 &\leq \sup_{t_{0} \leq t \leq T}   e^{\underline{w}r_{t}} + \sup_{t_{0} \leq t \leq T}  e^{\overline{w}r_{t}},
\end{align*}
where
\[
\underline{w}=\min_{t_{0} \leq t \leq T} w(t)\qquad\text{and}\qquad \overline{w}=\max_{t_{0} \leq t \leq T} w(t).
\]
Thus, using H\"{o}lder inequality and (\ref{one_1}), we obtain
\[
\mathbb{E}_{r_{0},t_{0}}^{\eta} \left[\sup_{t_{0} \leq t \leq T} e^{w(t)r_{t}}\right]\leq \sqrt{\mathbb{E}^{P} \left(  \frac{dQ^{\eta}}{dP} \right)^{2}}  \sqrt{\mathbb{E}_{r_{0},t_{0}}^{P} \left[\sup_{t_{0} \leq t \leq T} e^{2w(t)r_{t}}\right]}<+\infty \qquad \forall\eta\in\mathcal{M}.
\]

Finally, if we divide equation (\ref{3}) by  $H(r_{t},t)$ and remembering that
\[
\frac{1}{H(r,t)}=-e^{-(T-t)r}F(r,t),
\]
$F$ is bounded, $G$ and $Y^{\eta^{*}}_{t}$ satisfy (\ref{4}) and
\[
\mathbb{E}_{r_{0},t_{0}}^{\eta} \left[\sup_{t_{0} \leq t \leq T} e^{-(T-t)r_{t}}\right]<+\infty\qquad \forall\eta\in\mathcal{M},
\]
 we get that (\ref{2}) holds.  
\end{proof}

\section{Example - Vasicek model}\label{sec:7} 

In this section we solve problem (\ref{problem1}) in the Vasicek model. Namely, we assume that $\sigma(t)>0$ is a continuous function, $\alpha>0$, $\lambda$, $\bar{\theta}$, $\bar{\sigma}\in\mathbb{R}$ and
\[
\lambda(r)=\lambda,\qquad\sigma(r,t)=\sigma(t),\qquad \bar{\mu}(r)=\bar{\theta}-\alpha r,\qquad \bar{\sigma}(r)=\bar{\sigma}.
\]
Since $\bar{\mu}$ is not bounded, the solution is not a straightforward consequence of Theorem \ref{final} and needs a separate proof. 
\begin{theorem}
Suppose that initial conditions $(x_{0},y_{0},r_{0},t_{0})$ are fixed. In the Vasicek model the saddle point
\begin{equation}\label{saddle_point_vasicek}
\left(\pi^{*},\eta^{*}(\pi^{*})\right)\in\mathcal{A}_{x_{0},y_{0},r_{0},t_{0}}\times\mathcal{M}
\end{equation}
looks as follows
\begin{align*}
\pi^{*}&=\frac{2Y_{t}^{\eta^{*}}}{\sigma(t)}\left(\lambda+\bar{\sigma}B_{1}(t)\right)\frac{A_{2}(t)}{A_{1}(t)}e^{-B_{1}(t)r_{t}}-X_{t}^{\pi^{*}}\frac{\bar{\sigma}}{\sigma(t)}B_{1}(t),\\
\eta^{*}(\pi^{*})&=-\lambda-\bar{\sigma}B_{1}(t),
\end{align*}
where
\begin{align}
B_{1}(t)&=\frac{1}{\alpha}\left(1-e^{-\alpha(T-t)}\right),\label{B1}\\
A_{1}(t)&=-\exp\left\{\int_{t}^{T}\left(\bar{\theta}-\bar{\sigma}\lambda\right)B_{1}(s)-\frac{1}{2}\bar{\sigma}^{2}B^{2}_{1}(s) \ ds\right\},\label{A1}\\
A_{2}(t)&=-\exp\left\{\int_{t}^{T}\left(\lambda+\bar{\sigma}B_{1}(s)\right)^{2} ds\right\}.\label{A2}
\end{align}
\end{theorem}

\begin{proof}
Let
\begin{equation}\label{22}
H(r,t)=A_{1}(t)e^{B_{1}(t)r},\qquad \text{where}\qquad A_{1}(T)=-1,\qquad B_{1}(T)=0.
\end{equation}
Then equation (\ref{first_equation}) is of the form
\[
A'_{1}(t)+rA_{1}(t)\left[B'_{1}(t)-\alpha B_{1}(t)+1\right]+\left(\bar{\theta}-\bar{\sigma}\lambda\right)A_{1}(t)B_{1}(t)-\frac{1}{2}\bar{\sigma}^{2}A_{1}(t)B_{1}^{2}(t)=0.
\]
If $B_{1}(t)$ looks as (\ref{B1}), then we get
\[
A'_{1}(t)+A_{1}(t)\left[\left(\bar{\theta}-\bar{\sigma}\lambda\right)B_{1}(t)-\frac{1}{2}\bar{\sigma}^{2}B_{1}^{2}(t)\right]=0,
\]
so $A_{1}(t)$ is of the form (\ref{A1}). Now using (\ref{22}) we can rewrite equation (\ref{second_equation}) as
\begin{equation}\label{33}
G_{t}+\frac{1}{2}\bar{\sigma}^{2}G_{rr}+\left(\lambda+\bar{\sigma}B_{1}(t)\right)^{2}G+\left(\bar{\theta}-\alpha r-2\bar{\sigma}\lambda-2\bar{\sigma}^{2}B_{1}(t)\right)G_{r}=0.
\end{equation}
Let
\begin{equation}\label{44}
G(r,t)=A_{2}(t),\qquad \text{where}\qquad A_{2}(T)=-1.
\end{equation}
Then equation (\ref{33}) looks as 
\[
A'_{2}(t)+A_{2}(t)\left(\lambda+\bar{\sigma}B_{1}(t)\right)^{2}=0,
\]
so $A_{2}(t)$ is of the form (\ref{A2}). Finally, saddle point candidate (\ref{saddle_point_candidate}) for $H$ and $G$ given by (\ref{22}) and (\ref{44}) respectively looks as (\ref{saddle_point_vasicek}). 

Now we can set 
\[
 V(x,y,r,t):=H(r,t)x+G(r,t)y,
\]
and check whether $\left(\pi^{*},\eta^{*}(\pi^{*})\right)$ belongs to the set $\mathcal{A}_{x_{0},y_{0},r_{0},t_{0}} \times\mathcal{M}$ and condition (\ref{uniform}) holds. Taking into account the form of $H$ and $G$
and the fact that
\[
H(r_{t},t)X^{\pi^{*}}_{t}=2G(r_{0},t_{0})y_{0}+H(r_{0},t_{0})x_{0}-2G(r_{t},t)Y^{\eta^{*}}_{t},
\]
it is sufficient to prove only that for any deterministic continuous function $w(t)$ and for any $\eta\in\mathcal{M}$ we have
\[
 \mathbb{E}_{r_{0},t_{0}}^{\eta}\left[\sup_{t_{0} \leq t \leq T} e^{w(t)r_{t}}\right]  < + \infty,
\]
where $(r_{t},t_{0}\leq t\leq T)$ is a stochastic process given by
\[
dr_{t}=(\bar{\theta}-\alpha r_{t})dt + \bar{\sigma}dW_{t}.
\]
If we define 
\[
\delta_{t}:=e^{\alpha t} r_{t},
\]
then 
\[
d\delta_{t}=\bar{\theta} e^{\alpha t} dt + \bar{\sigma} e^{\alpha t} dW_{t},
\]
so 
\[
\delta_{t}= e^{\alpha t_{0}} r_{t_{0}} +  \int_{t_{0}}^{t}\bar{\theta} e^{\alpha s} ds + \int_{t_{0}}^{t}\bar{\sigma}e^{\alpha s} dW_{s}
\]
and note that
\begin{equation}\label{warunek_a}
 \mathbb{E}_{r_{0},t_{0}}^{P}\left[\sup_{t_{0} \leq t \leq T} e^{\delta_{t}}\right]<+\infty.
\end{equation}
Using the same method as in proof of Theorem \ref{final} we obtain
\begin{equation}\label{ccc}
 \sup_{t_{0} \leq t \leq T} e^{w(t)\delta_{t}}  \leq \sup_{t_{0} \leq t \leq T}   e^{\underline{w}\delta_{t}} + \sup_{t_{0} \leq t \leq T}  e^{\overline{w}\delta_{t}},
\end{equation}
where
\[
\underline{w}=\min_{t_{0} \leq t \leq T} w(t)\qquad\text{and}\qquad \overline{w}=\max_{t_{0} \leq t \leq T} w(t).
\]
Finally, taking into account (\ref{warunek_a}) and (\ref{ccc}), we get
\[
\mathbb{E}_{r_{0},t_{0}}^{\eta} \left[\sup_{t_{0} \leq t \leq T} e^{w(t)r_{t}}\right]= \mathbb{E}_{r_{0},t_{0}}^{\eta} \left[\sup_{t_{0} \leq t \leq T} \exp\left\{w(t)e^{-\alpha t}\delta_{t}\right\} \right] 
\]
\[
\leq \sqrt{\mathbb{E}^{P} \left(  \frac{dQ^{\eta}}{dP} \right)^{2}}  \sqrt{\mathbb{E}_{r_{0},t_{0}}^{P} \left[\sup_{t_{0} \leq t \leq T} \exp\left\{2w(t)e^{-\alpha t}\delta_{t}\right\}\right]}<+\infty\qquad \forall\eta\in\mathcal{M}.
\]
\end{proof}

\end{document}